\documentclass[12pt]{iopart}

\usepackage{graphicx}
\usepackage{bm}
\usepackage{amsfonts}
\usepackage{amscd}

\usepackage{enumerate}

\usepackage{amssymb}
\usepackage{amsthm}
\usepackage{amstext}
\usepackage{amsopn}

\usepackage{setstack}
\usepackage{iopams_add}

\newtheorem{theorem}{Theorem}
\newtheorem{lemma}{Lemma}
\newtheorem{corollary}{Corollary}

\newtheorem{definition}{Definition}

\newtheorem{observation}{Observation}

\newcommand{\bra}[1]{\mbox{$\left\langle #1 \right|$}}
\newcommand{\ket}[1]{\mbox{$\left| #1 \right\rangle$}}
\newcommand{\braket}[2]{\mbox{$\left\langle #1 | #2 \right\rangle$}}

\def\tr{{\rm Tr}}
\def\IR{{\mathbb R}}
\def\IC{{\mathbb C}}

\def\ChSum{{\mathbb T}}

\newcommand{\eqref}[1]{(\ref{#1})}
\newcommand{\implies}{\Longrightarrow}

\newcommand{\textcolor}[1]{}

\begin{document}

\title
{Disguising quantum channels by mixing and channel distance trade-off}

\author{Chi-Hang Fred Fung and H.~F. Chau}
\address{Department of Physics and Center of Theoretical and Computational Physics, University of Hong Kong, Pokfulam Road, Hong Kong}
\ead{chffung@hku.hk}


\begin{abstract}
We consider the reverse problem to the distinguishability of two quantum channels, which we call the disguising problem.
Given two quantum channels, 
the goal here is to make the two channels identical by mixing with some other channels with minimal mixing probabilities.
This quantifies
how much one channel can disguise as the other.
In addition,
the possibility to trade off between the two mixing probabilities 
allows one channel to be more preserved (less mixed) at the expense of the other.
We derive lower- and upper-bounds of the trade-off curve and apply them to a few example channels.
Optimal trade-off is obtained in one example.
We relate the disguising problem and the distinguishability problem by showing the the former can lower and upper bound the diamond norm.
We also show that the disguising problem gives an upper bound on the key generation rate in quantum cryptography.
\end{abstract}


\maketitle

\section{Introduction}

Quantum information processing involves the transformation of quantum states through quantum channels and 
it is often useful to quantify how far apart 
quantum states or quantum channels are.
Depending on the problem at hand, different ways of measuring the distance may be adopted.
Trace distance~\cite{Nielsen2000,Bengtsson2008} and fidelity~\cite{Uhlmann1976273,Alberti1983,Jozsa1994} are two widely-used measures for quantum states.
Trace distance is particularly interesting because it corresponds to a measurement that distinguishes between two quantum states with the minimum error.
Other distances for quantum states have also been studied recently, including the Monge distance~\cite{Zyczkowski:2001:Monge}, the $k$th operator norm~\cite{Johnston:2010}, and the partitioned trace distance~\cite{Rastegin:2010}.
For quantum channels,
measures~\cite{PhysRevA.71.062310} have also been proposed based on extending
the fidelity measure~\cite{Belavkin:2005} and the trace distance measure~\cite{Kitaev:1997} of quantum states.
The diamond norm, in particular, is a trace-distance-based measure for quantum channels.
It was first introduced in quantum information processing by Kitaev~\cite{Kitaev:1997} for studying quantum error correction and
has a nice operational meaning because it corresponds to minimum-error channel discrimination. 
As such,
the diamond norm has been receiving a lot of attention since its introduction, in both the theoretical aspect
\cite{Aharonov:1998:QCM:276698.276708,Rosgen:2005:distinguishing,Sacchi:2005:discrimination,Li2008_diamond,Piani:2009:discrimination,Yu2012:bounds}
and the computational aspect
\cite{Watrous:2009:SDP,Benenti:2010:approximation}.

While distinguishability (of quantum states and channels) is a well studied problem,
we consider the reverse problem -- the disguising problem for quantum channels.
Unlike the distinguishability problem in which the goal is to find a measurement that distinguishes between two (or more) states or channels, the aim in the disguising problem is to find out the minimal mixing needed to make two (or more) quantum channels completely identical.
In essence, this quantifies 
how much the effect of one channel is partially carried out by another channel. 
In this pilot study, we investigate the disguising problem for two channels.

As we show in this paper, 
the disguising problem and the distinguishability problem can be considered as dual to each other.
We establish this 
by showing that the solution of the disguising problem can 
be used to lower and upper bound the diamond norm,
{\textcolor{mycolor2}{
which is a measure of distinguishability.
}}%
This has an interesting implication:
the more distinguishable two quantum channels are, the more effort it takes to disguise one as the other.
Additionally, the disguising problem can be cast as a semidefinite program and we also show efficient ways to compute lower- and upper-bounds of it.
{\textcolor{mycolor2}{
We note that the diamond norm
can be computed using
semidefinite/convex programming and Monte Carlo methods~\cite{Johnston:2009,Watrous:2009:SDP,Ben-Aroya:2010,Benenti:2010:approximation}.
}}%

The disguising problem can be understood with the following operational interpretation.
First note that the operational meaning of the diamond norm is based on the perspective of the receiver who tries to distinguish between two channels.
A reverse perspective is to look at the channel intervener who tries to make the channels identical by minimal intervention.
The channel intervener possesses the two original channels as black boxes.
She is not allowed to open them and 
is only allowed
to occasionally substitute each of them with some other arbitrary channel.
We ask what are the minimal 
mixing probabilities needed to make the two intervened channels identical?

The precise problem statement is the following.
Given two quantum channels $\mathcal{E}(\rho)=\sum_i E_i \rho E_i^\dag$ and $\mathcal{F}(\rho)=\sum_i F_i \rho F_i^\dag$ 
acting on an $n$-dimensional Hilbert space
where $E_i$ and $F_i$ are $n \times n $ complex matrices representing the Kraus operators of the channels,
we consider the processing
\begin{eqnarray}
\mathcal{E}(\rho) \rightarrow \mathcal{E}'(\rho)&=&(1-p) \mathcal{E}(\rho) + p  \mathcal{E}_\Delta (\rho),
\label{eqn-mixed-channel-E}
\\
\mathcal{F}(\rho) \rightarrow \mathcal{F}'(\rho)&=&(1-q) \mathcal{F}(\rho) + q  \mathcal{F}_\Delta (\rho) 
\label{eqn-mixed-channel-F}
\end{eqnarray}
such that 
$\mathcal{E}'=\mathcal{F}'$ with 
{\textcolor{mycolor2}{
``the smallest'' $p$ and $q$ (which will be clarified later as the distance profile).
}}%
In other words,
we are interested in 
the least amount of 
substitution
of $\mathcal{E}$ and $\mathcal{F}$ needed 
to make them equal.
This is illustrated in figure~\ref{fig-make-channels-equiv}.
Operationally, 
the new channel probabilistically selects between the original channel and some other 
{\it harmonizing channel}, 
$\mathcal{E}_\Delta$ or $\mathcal{F}_\Delta$, which is yet to be determined.
The smaller the 
{\it mixing probability}, 
$p$ or $q$,
the closer the new channel is to the original one.
Thus, $p$ and $q$ serve as a distance between the two channels.
We note that in general the 
harmonizing channels 
$\mathcal{E}_\Delta$ and $\mathcal{F}_\Delta$ 
are not universal and
depend on the original channels $\mathcal{E}$ and $\mathcal{F}$.
In fact, the problem becomes trivial if we insist $\mathcal{E}_\Delta$ and $\mathcal{F}_\Delta$ to be universal for then $\mathcal{E}_\Delta=\mathcal{F}_\Delta$ (with either $p=q=1$ when $\mathcal{E}\neq\mathcal{F}$ or $p=q=0$ when $\mathcal{E}=\mathcal{F}$).
Also, note that $\mathcal{E}'=\mathcal{F}'$ is trivially satisfied with $\mathcal{E}_\Delta=\mathcal{F}$, $\mathcal{F}_\Delta=\mathcal{E}$, and $p+q=1$.
This gives a linear trade-off between $p$ and $q$.
However, in general, better sub-linear trade-off can be obtained, as we show later.
Note that $\mathcal{E}_\Delta$ and $\mathcal{F}_\Delta$ can be general quantum channels with arbitrary complexity.

\begin{figure}
\begin{center}
\includegraphics[width=.5\columnwidth]{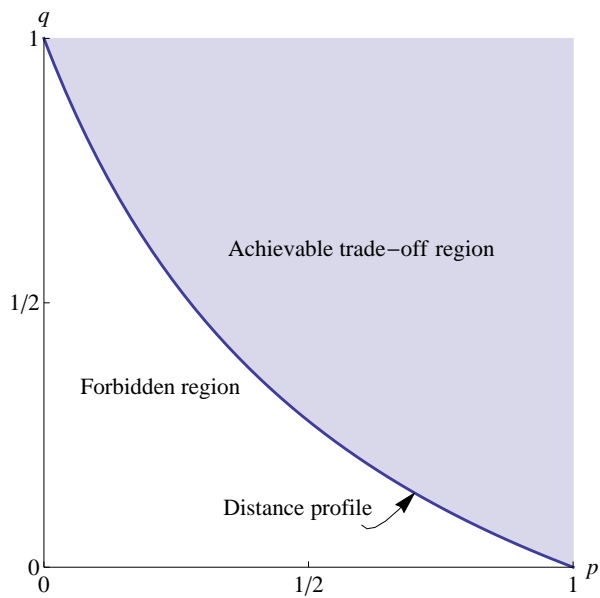}
\caption{
\label{fig:distance-profile}
Achievable trade-off region and distance profile of the channel disguising problem.
When the mixing probabilities $p$ and $q$ are too small (corresponding to the forbidden region), the channels $\mathcal E$ and $\mathcal F$ cannot be made equal.
}
\end{center}
\end{figure}

The pair of parameters $(p,q)$ represents the trade-off between the two channels'
 mixing probabilities.
The more 
mixing
is imposed on one channel, the less 
mixing
is required on the other.
Note that if $(p,q)$ is a trade-off point, $(p+\delta p, q+\delta q)$ with $\delta p, \delta q \geq 0$ is a also a trade-off point (and we call the former point strictly better than the latter).
When given a region of achievable trade-off points, a trade-off curve can be obtained by tracing out the boundary such that no point is strictly better than another.
This gives us a distance profile for the two channels 
{\textcolor{mycolor2}{
(see figure~\ref{fig:distance-profile}).
}}%
Thus, our measure is unique in that it is represented by a 2-dimensional curve rather than a scalar as in other measures for quantum channels.
On the other hand,
a scalar distance may be obtained from our measure 
in several ways,
for example,
(i) by
imposing
equal 
mixing
probabilities 
$p=q$ and regarding the minimum $p=q$ as the distance between the two channels,
or (ii) by regarding the minimum $p+q$ as the distance.
We will justify that these two are distances by showing that the triangle inequality holds.


\begin{figure*}
\begin{center}
\includegraphics{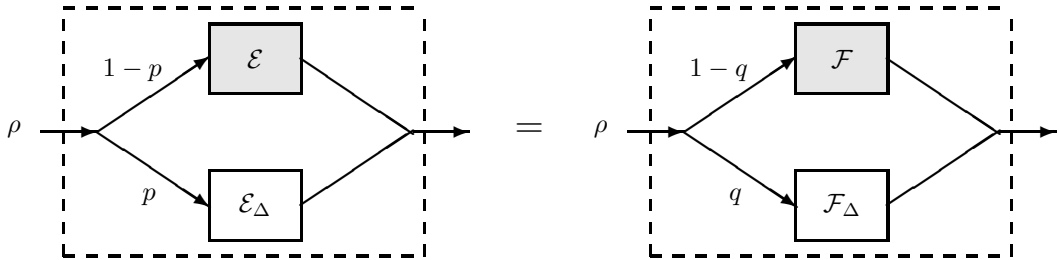}
\caption{
\label{fig-make-channels-equiv}
Two quantum channels, $\mathcal{E}$ and $\mathcal{F}$, are made identical by
mixing the original channel $\mathcal{E}$ ($\mathcal{F}$) with a 
harmonizing channel
$\mathcal{E}_\Delta$ ($\mathcal{F}_\Delta$) with probability $p$ ($q$).
The designer is free to choose the non-shaded parts $\mathcal{E}_\Delta$ and $\mathcal{F}_\Delta$, 
but not $\mathcal{E}$ and $\mathcal{F}$.
}
\end{center}
\end{figure*}


The disguising problem admits a geometric interpretation.
Given a channel $\mathcal{E}$, we denote the set of all channels achieved by mixing channel $\mathcal{E}$ with arbitrary harmonizing channels and mixing probability $p$ as
\begin{eqnarray*}
S_p(\mathcal{E}) 
&\equiv
&
\{
\mathcal{E}':
\mathcal{E}'=(1-p) \mathcal{E} + p  \mathcal{E}_\Delta ,
\nonumber
\text{ for some harmonizing channel } \mathcal{E}_\Delta
\} .
\end{eqnarray*}
Note that $S_{p'} \subset S_p$ for $p' < p$ for the following reason.
For any $\mathcal{E}' \in S_{p'}$, we have
\begin{eqnarray*}
\mathcal{E}'&=&(1-p') \mathcal{E} + p'  \mathcal{E}_\Delta
\\
&=&(1-p) \mathcal{E} + [p'  \mathcal{E}_\Delta + (p-p') \mathcal{E} ]
\end{eqnarray*}
where the term in bracket is a valid quantum channel scaled by $p$.
Thus, $\mathcal{E}'$ 
 can be regarded as having mixing probability $p$ and
so $\mathcal{E}' \in S_{p}$.
It can be easily checked that $S_p(\mathcal{E})$ is compact and convex.
This enables a geometric interpretation of the disguising problem as a search for $p$ and $q$ so that $S_p(\mathcal{E})$ and $S_q(\mathcal{F})$ just meet (see figure~\ref{fig:geometric}).
\begin{figure}
\begin{center}
\includegraphics[width=.5\columnwidth]{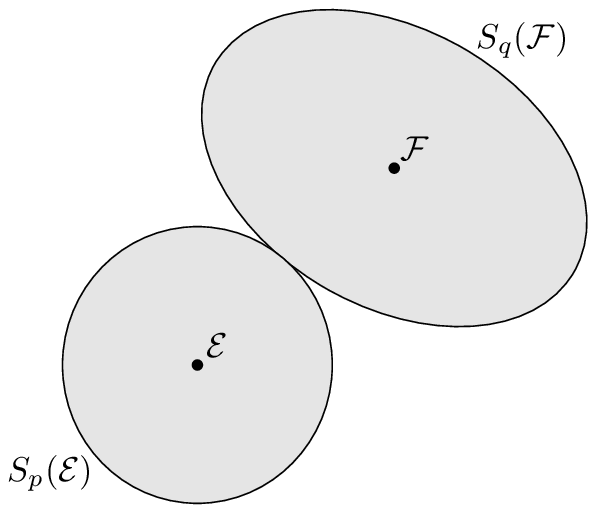}
\centerline{(a)}
\includegraphics[width=.5\columnwidth]{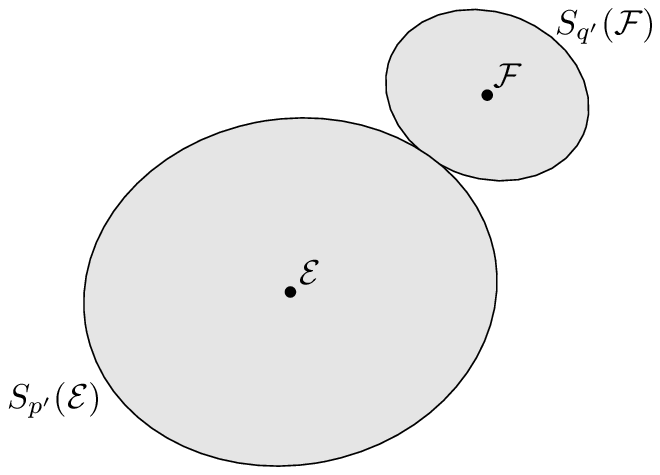}
\centerline{(b)}
\caption{
\label{fig:geometric}
Geometric interpretation of the disguising problem.
The set of channels within mixing probability $p$ ($q$) from channel $\mathcal{E}$ ($\mathcal{F}$) is denoted as $S_p(\mathcal{E})$ ($S_q(\mathcal{F})$).
In (a), we search for $p$ and $q$ so that the two sets meet.
In (b), they meet with different parameters $p'$ and $q'$.
}
\end{center}
\end{figure}

In this paper, we formulate the disguising problem as an optimization problem and 
{\textcolor{mycolor2}{
present 
}}%
our main result in section~\ref{sec-problem-formulation}.
Although solving it turns out to be difficult, we are able to obtain lower-bound and upper-bound on the $(p,q)$ trade-off curve (cf. equation~\eqref{eqn-main-result-final-bounds}).
{\textcolor{mycolor2}{
In section~\ref{sec-proof}, we 
prove the main result, which is the lower- and upper-bounds.
}}%
Next, we illustrate the computation of the bounds in a few examples for different quantum channels in section~\ref{sec-examples}.
In one special case, the analytical lower- and upper-bounds coincide, effectively producing the optimal trade-off curve.
For the other cases, the numerically computed lower- and upper-bounds are quite tight, showing the effectiveness of the bounds.
In section~\ref{sec-relation-distinguishability},
we show that the disguising problem can lower and upper bound the diamond norm which quantifies the distinguishability of two quantum channels.
We also make three remarks about our distance profile for quantum channels in section~\ref{sec-remarks}.
We discuss one application of the disguising problem in section~\ref{sec-application}, which is to bound the key generation rate in quantum cryptography.
We conclude in section~\ref{sec-conclusions}.

\section{Notation}
We denote a matrix $A$ to be positive-semidefinite (PSD) by $A \succeq 0$, the transpose of $A$ by $A^t$, and the conjugate transpose of $A$ by $A^\dag$.
$A$ is PSD if and only if $A$ is Hermitian and its eigenvalues are non-negative.
$\lVert A \rVert$ denotes the spectral norm of $A$ which is the largest singular value of $A$ or the largest eigenvalue of $A$ if $A$ is PSD.
$\lVert A \rVert_1 \triangleq \tr \sqrt{A^\dag A}$ denotes the trace norm of $A$.

$\mathcal{B}(\mathcal{H}_n)$ denotes the set of all bounded linear operators 
in an $n$-dimensional Hilbert space $\mathcal{H}_n$.
$I_p$ denotes the identity operator in a $p$-dimensional Hilbert space.
A linear map $\mathcal{E}:\mathcal{B}(\mathcal{H}_n) \rightarrow \mathcal{B}(\mathcal{H}_n)$ is positive
if $\mathcal{E}(A)$ is PSD for all PSD $A$ in $\mathcal{B}(\mathcal{H}_n)$, and 
$\mathcal{E}$ is completely positive (CP) if $\mathcal{E} \otimes I_p$ is positive for all positive integers $p$.
$\mathcal{E}$ is trace-preserving (TP) if 
$\tr(\mathcal{E}(A))=\tr(A)$ for all $A$ in $\mathcal{B}(\mathcal{H}_n)$.

A linear map 
$\mathcal{E}:\mathcal{B}(\mathcal{H}_n) \rightarrow \mathcal{B}(\mathcal{H}_n)$ 
can be represented by
a Choi matrix of size $n^2 \times n^2$~\cite{Choi1975:CPMap}:
\begin{eqnarray}
C_\mathcal{E}&=&
\begin{bmatrix}
\mathcal{E}(\ket{0}\bra{0}) & \mathcal{E}(\ket{0}\bra{1}) & \dots & \mathcal{E}(\ket{0}\bra{n-1})\\
\mathcal{E}(\ket{1}\bra{0}) & \ddots\\
\vdots\\
\mathcal{E}(\ket{n-1}\bra{0}) &  \dots & & \mathcal{E}(\ket{n-1}\bra{n-1})
\end{bmatrix}
\label{eqn-choi-matrix}
\\
&=&\sum_{i,j=0}^{n-1} \ket{i}\bra{j} \otimes \mathcal{E}(\ket{i}\bra{j}).
\end{eqnarray}
We define a function, which we call the {\it channel sum} function, of the Choi matrix $C_\mathcal{E}$ of a linear map $\mathcal{E}$ as follows:
\begin{eqnarray}
\label{eqn-channel-sum-def}
\ChSum(C_\mathcal{E})
&:=&
\tr_2^t ( C_\mathcal{E} )
\\
&=&
\sum_{i,j=0}^{n-1} \ket{i}\bra{j} \cdot \tr[\mathcal{E}(\ket{j}\bra{i})] ,
\end{eqnarray}
where $\tr_2$ is the partial trace over the second system and $t$ represents transpose.
We remark that 
$\ChSum(C_\mathcal{E})=I$ if and only if $\mathcal{E}$ is trace-preserving (see Lemma~\ref{lemma-TP} in 
\ref{app-preliminary}).

\section{Problem formulation and main result}
\label{sec-problem-formulation}

\subsection{Optimization problem formulation}
\label{sec-problem-formulation-detail}

To solve for the optimal distance profile $(p,q)$ of
equations~\eqref{eqn-mixed-channel-E} and \eqref{eqn-mixed-channel-F}
with the condition that the new channels are identical, i.e.,
$\mathcal{E}'=\mathcal{F}'$,
we formulate the problem as 
{\textcolor{mycolor2}{
(see figure~\ref{fig:distance-profile}):
\begin{eqnarray}
\label{eqn-original-minimization-prob-0}
&\text{minimize } & \: 
q\\
&\text{subject to } & \mathcal{E}'=\mathcal{F}' , \nonumber\\
&&
C_{\mathcal{E}_\Delta} \succeq 0 , \nonumber\\
&&
C_{\mathcal{F}_\Delta} \succeq 0 , \nonumber\\
&&
\ChSum(C_{\mathcal{E}_\Delta})=I , \nonumber\\
&&
\ChSum(C_{\mathcal{F}_\Delta})=I , \nonumber
\end{eqnarray}
where 
the minimization is over
$q$, $\mathcal{E}_\Delta$, and $\mathcal{F}_\Delta$, for some fixed $p$.
}}%
Here, we denote the Choi matrices of $\mathcal{E}_\Delta$ and $\mathcal{F}_\Delta$
by $C_{\mathcal{E}_\Delta}$ and $C_{\mathcal{F}_\Delta}$, respectively (see equation~\eqref{eqn-choi-matrix}), and
$\ChSum$ is the channel sum function
defined in 
equation~\eqref{eqn-channel-sum-def}.
The last four constraints demand that $\mathcal{E}_\Delta$ and $\mathcal{F}_\Delta$ be quantum channels (TPCP maps) (cf. Theorem~\ref{thm-Choi-CP} and 
Lemma~\ref{lemma-TP} in 
\ref{app-preliminary}).
{\textcolor{mycolor2}{
Note that the roles of $p$ and $q$ in the formulation of the above optimization problem may be interchanged (i.e., we may have ``fix $q$ and minimize $p$'' instead).
}}%
We remark that this problem can be cast as a semidefinite program with the use of equation~\eqref{eqn-channels-difference0} and may be solved numerically.
However, in this paper, we are interested in the analytical bounds of this problem and the investigation of the trade-off behavior between $p$ and $q$.

Note that for $p=1$, the solution of $q=0$ is trivially obtained to make 
$\mathcal{E}'=\mathcal{F}'$, since 
we can choose $\mathcal{E}_\Delta=\mathcal{F}$ in
equations~\eqref{eqn-mixed-channel-E} and \eqref{eqn-mixed-channel-F}.
By the same token, $(p,q)=(0,1)$ is feasible.

The distance profile (such as that in 
{\textcolor{mycolor2}{%
figure~\ref{fig:distance-profile}%
}}%
) should be convex.
This is because
given two points $(p,q)$ and $(p',q')$ that 
{\textcolor{mycolor2}{%
satisfy 
}}%
$\mathcal{E}'=\mathcal{F}'$ [see equations~\eqref{eqn-mixed-channel-E} and \eqref{eqn-mixed-channel-F}],
any linear combination of them [i.e., $((1-t)p+t p',(1-t)q + t q')$ for some $0\leq t \leq 1$] also satisfies it.
It follows that any point on the
line $q=1-p$ is a feasible solution.

\begin{figure}
\begin{center}
\includegraphics[width=.5\columnwidth]{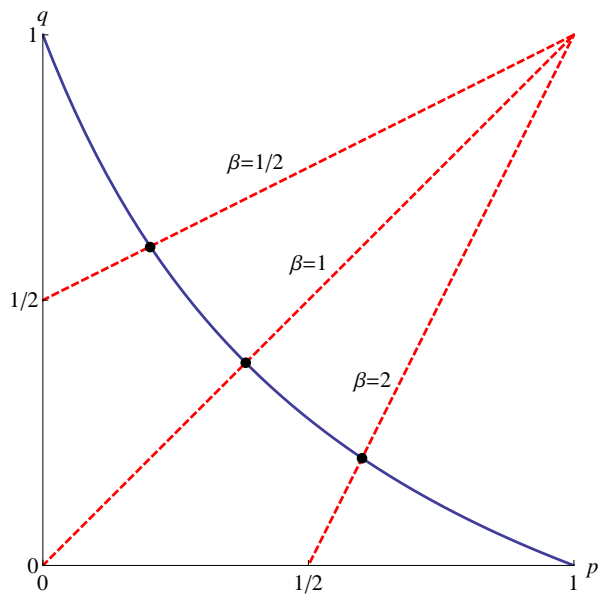}
\caption{
\label{fig:optimization-prob}
Distance profile obtained by 
solving problem~\eqref{eqn-original-minimization-prob}.
Among all the points $(p,q)$ along the line corresponding to a fixed value of $\beta$, we choose the point with the minimum $q$ subject to the constraints of the problem.
For each value of $\beta$, we solve the optimization problem and obtain an optimal point $(p,q)$.
Repeating this for a range of $\beta$ produces the solid curve.
}
\end{center}
\end{figure}

\subsection{Main result}

{\textcolor{mycolor2}{
Our main result is the lower- and upper-bounds of the distance profile generated by 
the solutions of problem~\eqref{eqn-original-minimization-prob-0}.
It turns out that it is easier to analyze and present the main result by expressing the problem as follows:
\begin{eqnarray}
\label{eqn-original-minimization-prob}
&\text{minimize } & \: 
q\\
&\text{subject to } & \mathcal{E}'=\mathcal{F}' , \nonumber\\
&&
\beta=\frac{1-q}{1-p} , \nonumber\\
&&
C_{\mathcal{E}_\Delta} \succeq 0 , \nonumber\\
&&
C_{\mathcal{F}_\Delta} \succeq 0 , \nonumber\\
&&
\ChSum(C_{\mathcal{E}_\Delta})=I , \nonumber\\
&&
\ChSum(C_{\mathcal{F}_\Delta})=I , \nonumber
\end{eqnarray}
where 
the minimization is over
$p$, 
$q$, $\mathcal{E}_\Delta$, and $\mathcal{F}_\Delta$, for some fixed 
parameter $\beta$.  
Here, $\beta$ is a new parameter and
the distance profile $(p,q)$ is obtained by solving this optimization problem over a range of $\beta$
(see figure~\ref{fig:optimization-prob}).
}}

{\textcolor{mycolor2}{
Our main result is that the distance profile $(p,q)$ is bounded from below and above as follows:
\begin{eqnarray}
\label{eqn-main-result-p}
p&=&1-\frac{1}{\alpha+\beta}, 
\text{ and}
\\
\label{eqn-main-result-q}
q&=&\frac{\alpha}{\alpha+\beta} 
\end{eqnarray}
where $\beta$ is fixed and
\begin{eqnarray}
\label{eqn-main-result-final-bounds}
n^{-1}\tr[\ChSum(\Delta_+)]
\leq
\alpha
\leq
\min(
\lVert \ChSum(\Delta_+)\rVert ,
1) .
\end{eqnarray}
Here, $\Delta_+$ is a PSD matrix obtained by
decomposing $C_\mathcal{E}-\beta C_\mathcal{F}$
into the positive and negative subspaces by eigen-decomposition:
\begin{equation}
C_\mathcal{E}-\beta C_\mathcal{F}=\Delta_+ - \Delta_-,
\end{equation}
where $\Delta_\pm$ are 
PSD matrices with support on orthogonal vector spaces
(i.e., $\Delta_+ \Delta_-=0$) and
$C_{\mathcal{E}}$ and $C_{\mathcal{F}}$
are the Choi matrices for the quantum channels $\mathcal{E}$ and $\mathcal{F}$, respectively.
Also, $n$ is the dimension of the quantum states on which the channels act.
}}

Note that $dp/d\alpha >0$ and $dq/d\alpha >0$ implying that a smaller $\alpha$ gives rise to a ``smaller'' pair $(p,q)$ in the 2-dimensional space.
This means that the lower (upper) bounds of $\alpha$ in
equation~\eqref{eqn-main-result-final-bounds} obtained by varying $\beta$ correspond to a lower (upper) bound curve in the $(p,q)$ space.

Note that if the lower bound and upper bound of $\alpha$ coincide, the optimal $\alpha$ and thus optimal $(p,q)$ are obtained.
This happens if and only if 
$\Delta_+$ already corresponds to a scaled quantum channel, i.e., 
$\ChSum(\Delta_+)=\alpha I$, which is not the case in general.
Also, note that 
equation~\eqref{eqn-main-result-final-bounds} implies that
the lower bound of $\alpha$ is always less than or equal to $1$, and thus if 
$\ChSum(\Delta_+)=\alpha I$, $\alpha \leq 1$.

\section{Proof of lower- and upper-bounds}
\label{sec-proof}

As noted earlier, the case of $p=0$ is trivial and thus we focus on the case $p<1$ in the following.
We now
analyze 
problem~\eqref{eqn-original-minimization-prob},
and as we will show later,
directly solving this problem turns out to be difficult.
Let us first focus on the condition that we want: $\mathcal{E}'=\mathcal{F}'$. 
By Theorem~\ref{thm-equiv-channels} in 
\ref{app-preliminary}, we convert this condition to the Choi-matrix equivalence 
$C_{\mathcal{E}'}=C_{\mathcal{F}'}$, which implies that
\begin{eqnarray}
(1-p)C_\mathcal{E}+p C_{\mathcal{E}_\Delta}
&=
(1-q)C_\mathcal{F}+q C_{\mathcal{F}_\Delta}
\label{eqn-channels-difference0}
\\
C_\mathcal{E}-\beta C_\mathcal{F}&=
\frac{q}{1-p} C_{\mathcal{F}_\Delta}-
\frac{p}{1-p}C_{\mathcal{E}_\Delta} 
\label{eqn-channels-difference1}
\end{eqnarray}
where $\beta=\frac{1-q}{1-p}$.
We decompose the left-hand side into the positive and negative subspaces by eigen-decomposition:
\begin{equation}
\label{eqn-channel-decomposition1}
C_\mathcal{E}-\beta C_\mathcal{F}=\Delta_+ - \Delta_-,
\end{equation}
where $\Delta_\pm$ are 
positive semidefinite matrices with support on orthogonal vector spaces
(i.e., $\Delta_+ \Delta_-=0$).
As such, by 
Theorem~\ref{thm-Choi-CP} in 
\ref{app-preliminary},
$\Delta_\pm$ correspond to some CP maps.

Note that 
$G_\mathcal{E}$, $G_\mathcal{F}$, $G_{\mathcal{F}_\Delta}$, $G_{\mathcal{E}_\Delta}$, $\Delta_+$, and $\Delta_-$ are all Choi matrices.

Comparing equations~\eqref{eqn-channels-difference1} and \eqref{eqn-channel-decomposition1},
since the positive and negative parts on the right-hand sides must match,
the Choi matrices of the 
harmonizing channels
must be of the form
\begin{eqnarray}
\label{eqn-delta-channels1-F}
\frac{q}{1-p} C_{\mathcal{F}_\Delta} &=& 
\Delta_+ + X ,
\\
\label{eqn-delta-channels1-E}
\frac{p}{1-p} C_{\mathcal{E}_\Delta} &=& 
\Delta_- + X
\end{eqnarray}
where $X$ is some Hermitian matrix corresponding to the Choi matrix of some linear map.
Note that $\Delta_\pm$
may not correspond to scaled quantum channels because 
$\ChSum(\Delta_\pm) \neq \alpha I$
for 
{\textcolor{mycolor2}{
any 
}}%
$\alpha >0$.
The purpose of adding $X$ is to make them scaled quantum channels
so that $\ChSum(\Delta_\pm + X) = \alpha_\pm I$.
\begin{lemma}
\label{lemma-SDp-SDm}
$\ChSum(\Delta_+)=\ChSum(\Delta_-)+(1-\beta)I
$.
\end{lemma}
\begin{proof}
Rearranging equation~\eqref{eqn-channel-decomposition1} and applying Corollary~\ref{cor-S-linear}, we have
\begin{eqnarray*}
&&
\ChSum(C_\mathcal{E}+\Delta_-)=
\ChSum(\beta C_\mathcal{F}+\Delta_+)\\
&\implies &
\ChSum(C_\mathcal{E})+\ChSum(\Delta_-)=
\ChSum(\beta C_\mathcal{F})+\ChSum(\Delta_+)\\
&\implies &
(1-\beta)I+
\ChSum(\Delta_-)=
\ChSum(\Delta_+)
\end{eqnarray*}
where we have used 
the fact that
$\ChSum(C_{\mathcal{F}_\Delta})=\ChSum(C_{\mathcal{E}_\Delta})=I$
since $C_{\mathcal{F}_\Delta}$ and $C_{\mathcal{E}_\Delta}$ are Choi matrices of quantum channels.
\end{proof}
As a consequence, $\ChSum(\Delta_+ + X) = \alpha I$ if and only if
$\ChSum(\Delta_- + X) = (\alpha+\beta-1) I$.
Furthermore, from equation~\eqref{eqn-delta-channels1-F},
since $\ChSum(C_{\mathcal{F}_\Delta})=I$, we have
\begin{equation}
\label{eqn-relation-p-alpha}
\frac{q}{1-p} = \alpha.
\end{equation}
The same expression is obtained when we consider equation~\eqref{eqn-delta-channels1-E} with $\ChSum(C_{\mathcal{E}_\Delta})=I$.
Thus, minimizing $q$ given $\beta$ fixed is equivalent to minimizing $\alpha$ given $\beta$ fixed, since 
\begin{equation}
\label{eqn-relation-alpha-beta}
\alpha=\frac{q}{1-q} \beta
\end{equation}
is an increasing function of $q$.
The original problem~\eqref{eqn-original-minimization-prob} becomes
\begin{eqnarray}
\hat{\alpha}=&\text{minimize } &\alpha \label{eqn-problem-min-alpha}\\
&\text{subject to }&
\Delta_+ + X \succeq 0 , \nonumber \\
&&\Delta_- + X \succeq 0 , \nonumber \\
&&\ChSum(\Delta_+ + X) = \alpha I , \nonumber \\
&&\ChSum(\Delta_- + X) = (\alpha+\beta-1) I , \nonumber
\end{eqnarray}
where the minimization is over Hermitian matrix $X$ given $\beta$ fixed, and $\Delta_\pm$ are from equation~\eqref{eqn-channel-decomposition1}.
Note that the fourth constraint is redundant due to Lemma~\ref{lemma-SDp-SDm} and is shown only for completeness.
Once $\alpha$ is found, we can compute $p$ and $q$ from 
equations~\eqref{eqn-relation-p-alpha} and \eqref{eqn-relation-alpha-beta}.

We investigate the form of $X$.
Since $C_{\mathcal{F}_\Delta}$ and $C_{\mathcal{E}_\Delta}$ represent quantum channels, they are PSD.
This means that, according to equations~\eqref{eqn-delta-channels1-F} and \eqref{eqn-delta-channels1-E}, $\Delta_\pm + X$ are PSD.
However, this does not mean that $X$ is also PSD,
and this makes finding the optimal $X$ difficult.
Nevertheless, we have the following constraint on $X$ which helps us bound $\hat{\alpha}$.

\begin{lemma}
\label{lemma-non-neg-trace}
{\rm
The constraints of problem~\eqref{eqn-problem-min-alpha} implies
$\tr (X) \geq 0$.
}
\end{lemma}
\begin{proof}
Since $\Delta_+$ and $\Delta_-$ are the positive and negative 
{\textcolor{mycolor2}{
ranges
}}%
of the matrix 
in equation~\eqref{eqn-channel-decomposition1}, we can identify non-overlapping projectors $P_+$ and $P_-$ onto them respectively.
We also define the projector onto the remaining subspace $P_0 = I-P_+ - P_-$.
Since $\Delta_\pm + X$ is PSD, we have
\begin{eqnarray*}
\tr [ P_- (\Delta_+ + X) ] &\geq& 0 , \\
\tr [ P_+ (\Delta_- + X) ] &\geq& 0 ,\text{ and}\\
\tr [ P_0 (\Delta_\pm + X) ] &\geq& 0 ,
\end{eqnarray*}
which implies that
\begin{eqnarray*}
\tr ( P_-  X ) &\geq& 0 ,\\
\tr ( P_+  X ) &\geq& 0 , \text{ and}\\
\tr ( P_0  X ) &\geq& 0.
\end{eqnarray*}
Summing these terms gives the desired result.
\end{proof}
This lemma implies that the non-zero eigenvalues of $X$ cannot be all negative, but $X$ can have positive and negative eigenvalues.

\begin{theorem}
{\rm
The optimal value of problem~\eqref{eqn-problem-min-alpha} is upper bounded by 
$\lVert \ChSum(\Delta_+)\rVert$.
}
\end{theorem}
\begin{proof}
To show an upper bound, we only need to find a feasible $X$.
Choose $M=\lVert \ChSum(\Delta_+)\rVert I - \ChSum(\Delta_+)$ as the difference between two channel sums.
Certainly, $M$ is PSD and thus 
can be written as $M=D_0^\dag D_0$ where $D_0$ is a square matrix.
$M$ represents the channel sum of the channel $\rho \rightarrow D_0 \rho D_0^\dag$.
Let $X=\ket{D_0}\bra{D_0}$ be the Choi representation of this channel where $\ket{D_0}$ is the vector form of $D_0$ (cf. equation~\eqref{eqn-vector-form}).
Since $X$ is PSD, $\Delta_\pm+X$ is PSD and the first two constraints of problem~\eqref{eqn-problem-min-alpha} are satisfied.
Note that $\ChSum(X)=M$ by construction (cf. Lemma~\ref{lemma-channel-sum-of-Choi-equiv} and Definition~\ref{def-channel-sum}).
Therefore, $\ChSum(\Delta_+ + X)=\ChSum(\Delta_+) + \ChSum(X)=\lVert \ChSum(\Delta_+)\rVert I$ by Corollary~\ref{cor-S-linear}.

Note that for this upper bound, we have chosen $X$ to be PSD.
\end{proof}

We computed $\lVert \ChSum(\Delta_+)\rVert$ for random quantum channels and found cases with $\lVert \ChSum(\Delta_+)\rVert > 1$.
Nevertheless, $\hat{\alpha} \leq 1$ is also a valid bound.

\begin{lemma}
{\rm
The optimal value of problem~\eqref{eqn-problem-min-alpha} is upper bounded by unity,
i.e., $\hat{\alpha} \leq 1$.
}
\end{lemma}
\begin{proof}
Set the 
harmonizing channels
in 
equations~\eqref{eqn-mixed-channel-E}--\eqref{eqn-mixed-channel-F}
to be
{\textcolor{mycolor2}{
$\mathcal{E}_\Delta =\mathcal{F}$ and
$\mathcal{F}_\Delta =\mathcal{E}$.
Then, $\mathcal{E}' =\mathcal{F}'$ 
}}%
is satisfied with 
$q=1-p$, which means that $\alpha=1$ according to equation~\eqref{eqn-relation-p-alpha}.
Based on equation~\eqref{eqn-channel-decomposition1}, 
we set 
$X=C_\mathcal{E}-\Delta_+=\beta C_\mathcal{F}-\Delta_-$.
Then, we have 
$\Delta_+ + X=C_\mathcal{E} \succeq 0$ and
$\Delta_- + X=\beta C_\mathcal{F} \succeq 0$.
As such, the constraints of problem~\eqref{eqn-problem-min-alpha} are satisfied with $\alpha=1$.
\end{proof}

We remark that in the above proofs of the two upper bounds, we have explicitly constructed $X$.
Therefore, problem~\eqref{eqn-problem-min-alpha} is always feasible.

\begin{theorem}
{\rm
The optimal value of problem~\eqref{eqn-problem-min-alpha} is lower bounded by 
$n^{-1} \tr [\ChSum(\Delta_+)]$.
}
\end{theorem}
\begin{proof}
The channel sum is $\ChSum(\Delta_+ + X)$, and
{\textcolor{mycolor2}{
the sum of the eigenvalues of the channel sum $\ChSum(\Delta_+ + X)$ is 
}}%
\begin{eqnarray*}
\tr[\ChSum(\Delta_+ + X)]&=&
\tr[\ChSum(\Delta_+)] + \tr[\ChSum(X)]
\\
&\geq&
\tr[\ChSum(\Delta_+)] ,
\end{eqnarray*}
where
the first line is due to
linearity of $\ChSum$ (cf. Corollary~\ref{cor-S-linear}) and
the second line is
due to Corollary~\ref{cor-same-trace} and Lemma~\ref{lemma-non-neg-trace} which imply $\tr[\ChSum(X)]=\tr (X) \geq 0$.
Finally, since $\ChSum(\Delta_+ + X)=\alpha I$ (cf. problem~\eqref{eqn-problem-min-alpha}), we have $\alpha \geq n^{-1}\tr[\ChSum(\Delta_+)]$.

\end{proof}

In summary, the solution of problem~\eqref{eqn-problem-min-alpha} is bounded as follows:
\begin{equation}
\label{eqn-final-bounds}
n^{-1}\tr[\ChSum(\Delta_+)]
\leq
\hat{\alpha}
\leq
\min(
\lVert \ChSum(\Delta_+)\rVert ,
1).
\end{equation}

If $\Delta_+$ already corresponds to a scaled quantum channel, i.e., 
$\ChSum(\Delta_+)=\alpha I$ for some $\alpha$, 
then 
the optimal 
solution can be found:
$\hat{\alpha}=n^{-1}\tr[\ChSum(\Delta_+)]=\lVert \ChSum(\Delta_+)\rVert=\alpha$. 
In this case, $C_{\mathcal{F}_\Delta}=\alpha^{-1} \Delta_+$ 
and
$C_{\mathcal{E}_\Delta}=(\alpha+\beta-1)^{-1} \Delta_-$
can be found from equations~\eqref{eqn-delta-channels1-F} and~\eqref{eqn-delta-channels1-E} with $X=0$ and equation~\eqref{eqn-relation-p-alpha}.

\subsection{Procedure for computing the lower- and upper-bound $(p,q)$ curves}
\label{sec-procedure}
Suppose that we are given two quantum channels $\mathcal{E}$ and $\mathcal{F}$ of dimension $n$.
\begin{enumerate}
\item
Compute the Choi matrices $C_\mathcal{E}$ and $C_\mathcal{F}$ for the two channels using equation~\eqref{eqn-choi-matrix}.
\item
Fix $\beta$ in the range of $(0,\infty)$.
[Note that $\beta=0$ or $\beta=\infty$ corresponds to $q=1$ or $p=1$ respectively, and these are trivial cases because either $\mathcal{E}'$ or $\mathcal{F}'$ becomes arbitrary.]
\item
Eigen-decompose equation~\eqref{eqn-channel-decomposition1} to obtain $\Delta_\pm$.
\item
Compute the channel sum $\ChSum(\Delta_+)$ using equation~\eqref{eqn-channel-sum-def}.
\item
Compute the lower and upper bounds on $\hat{\alpha}$ using equation~\eqref{eqn-final-bounds}.
\item
Given a bound, denoted as $\alpha$, solve for $p$ and $q$ using 
equations~\eqref{eqn-main-result-p} and \eqref{eqn-main-result-q}.
\end{enumerate}
We can repeat this procedure for a range of $\beta$ to obtain the lower- and upper-bound $(p,q)$ trade-off curves.

\section{Examples}
\label{sec-examples}

\subsection{Difference between bit-flip and phase-flip channels}
\label{sec-example-flip-channels}

Given the bit-flip and phase-flip channels,
\begin{eqnarray}
\label{eqn-example-bit-flip}
\mathcal{E}(\rho) &=& (1-a) I_2 \rho I_2 + a X \rho X
\\
\label{eqn-example-phase-flip}
\mathcal{F}(\rho) &=& (1-b) I_2 \rho I_2 + b Z \rho Z
\end{eqnarray}
where $a$ and $b$ are the bit-flip and phase-flip probabilities, and
\begin{eqnarray*}
I_2=
\begin{bmatrix}
1 & 0
\\
0 & 1
\end{bmatrix}, \:
X=
\begin{bmatrix}
0 & 1
\\
1 & 0
\end{bmatrix}, \:
Z=
\begin{bmatrix}
1 & 0
\\
0 & -1
\end{bmatrix},
\end{eqnarray*}
we compute the Choi matrices for the two channels and find the difference
\begin{equation}
\label{eqn-example-bit-phase-Choi-diff}
\eqalign{
&C_\mathcal{E}-\beta C_\mathcal{F}
\cr
=&
(1-a-\beta+b \beta) \ket{e_1}\bra{e_1}
-
b \beta \ket{e_2}\bra{e_2}
+
a \ket{e_3}\bra{e_3}
}
\end{equation}
where $\bra{e_1}=[1, 0, 0, 1]$, $\bra{e_2}=[1, 0, 0, -1]$, and $\bra{e_3}=[0, 1, 1, 0]$.
Next, we separate this into the positive and negative subspaces as in equation~\eqref{eqn-channel-decomposition1}.
Note that since we consider $a,b,\beta > 0 $, the second term of the equation is negative and the third term is positive, while the first term can be non-negative or negative.

Case 1: $1-a-\beta+b \beta \geq 0$.  According to equation~\eqref{eqn-channel-decomposition1}, we have 
\begin{eqnarray*}
\Delta_+
&=&
(1-a-\beta+b \beta) \ket{e_1}\bra{e_1}
+
a \ket{e_3}\bra{e_3} ,
\\
\Delta_-
&=&
b \beta \ket{e_2}\bra{e_2}
\end{eqnarray*}
and
\begin{eqnarray*}
\ChSum(\Delta_+)&=&
(1-\beta+b \beta)I_2 ,
\\
\ChSum(\Delta_-)&=&
b \beta I_2 .
\end{eqnarray*}
Therefore, using the bounds in equation~\eqref{eqn-final-bounds}, we obtain the optimal solution of problem~\eqref{eqn-problem-min-alpha} as 
$\hat{\alpha}=1-\beta+b \beta$.

Case 2: $1-a-\beta+b \beta < 0$.  According to equation~\eqref{eqn-channel-decomposition1}, we have 
\begin{eqnarray*}
\Delta_+
&=&
a \ket{e_3}\bra{e_3} ,
\\
\Delta_-
&=&
-(1-a-\beta+b \beta) \ket{e_1}\bra{e_1}
+
b \beta \ket{e_2}\bra{e_2}
\end{eqnarray*}
and
\begin{eqnarray*}
\ChSum(\Delta_+)&=&
a I_2 ,
\\
\ChSum(\Delta_-)&=&
(-1+a+\beta) I_2 .
\end{eqnarray*}
Therefore, using the bounds in equation~\eqref{eqn-final-bounds}, we obtain the optimal solution of problem~\eqref{eqn-problem-min-alpha} as 
$\hat{\alpha}=a$.
Note that we are able to obtain the optimal solution in both cases instead of upper and lower bounds.

Finally, with $\alpha$ found for each case, we can compute a relation for $p$ and $q$ using equations~\eqref{eqn-relation-p-alpha} and \eqref{eqn-relation-alpha-beta}:
\begin{equation}
\label{eqn-example-final-relation-pq}
\left\{
\begin{array}{ll}
p=b - bq &, \quad \text{if } 1-a-\beta+b \beta \geq 0
\\
q=a-a p &, \quad \text{if } 1-a-\beta+b \beta < 0
\\
\end{array}
\right.
.
\end{equation}
This relation is depicted as the solid curve in figure~\ref{fig:example-bit-phase}, where the top-left (bottom-right) part corresponds to the first (second) case in equation~\eqref{eqn-example-final-relation-pq}.
Essentially, the cusp in the figure is due to the transition from case 1 with 2 positive and 1 negative eigenvalues to case 2 with 1 positive and 2 negative eigenvalues in equation~\eqref{eqn-example-bit-phase-Choi-diff}.
\begin{figure}
\begin{center}
\includegraphics[width=.6\columnwidth]{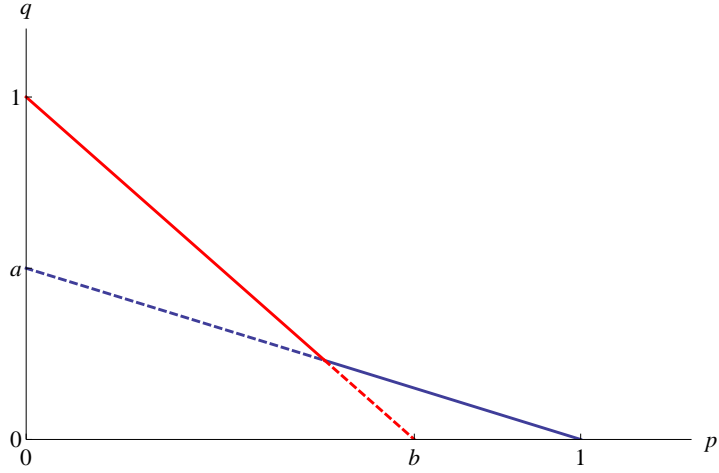}
\caption{
\label{fig:example-bit-phase}
Optimal trade-off curve for the 
mixing
probabilities $p$ and $q$ [defined in equations~\eqref{eqn-mixed-channel-E} and \eqref{eqn-mixed-channel-F}] for the bit- and phase-flip channels given in equations~\eqref{eqn-example-bit-flip} and \eqref{eqn-example-phase-flip}.
This curve is the solution to problem~\eqref{eqn-original-minimization-prob} or problem~\eqref{eqn-problem-min-alpha}.
}
\end{center}
\end{figure}

\begin{figure}
\begin{center}
\includegraphics[width=.6\columnwidth]{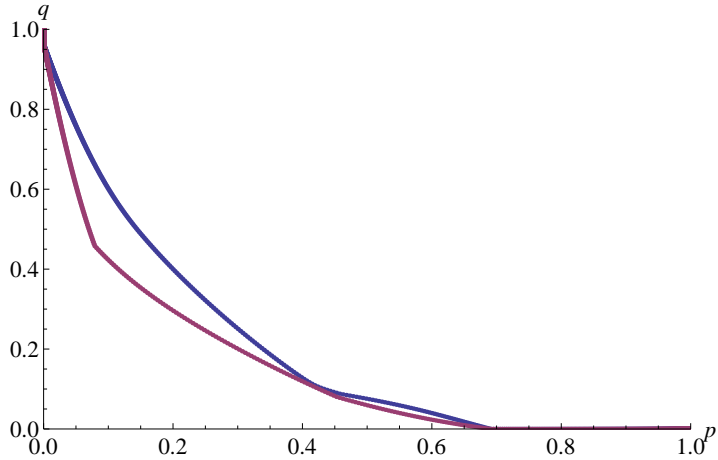}
\caption{
\label{fig:example-qubit-channels}
Lower- (purple) and upper-bound (blue) curves 
for the 
mixing
probabilities $p$ and $q$ [defined in equations~\eqref{eqn-mixed-channel-E} and \eqref{eqn-mixed-channel-F}] for two random qubit channels.
}
\end{center}
\end{figure}

\subsection{A pair of random qubit channels}
\label{sec-example-random-qubit-channels}

We randomly generated two qubits channels each having four Kraus operators and they are listed in 
\ref{app-example-random-qubit-channels}.
Using the procedure given in section~\ref{sec-procedure}, we compute the lower- and upper-bound curves which are shown in figure~\ref{fig:example-qubit-channels}.
We make two observations.
First, four cusps are obvious in the lower-bound curve, which are due to the transition of an eigenvalue of equation~\eqref{eqn-channel-decomposition1} from positive to negative (or vice versa).
Note that at most four cusps can occur since the dimension of the Choi matrices are four.
Second, there are regions where $q=0$ for a range of $p$ and where $p=0$ for a range of $q$ (the former is much bigger than the latter).
These regions correspond to the case that one channel contains another channel and we will clarify this concept in section~\ref{sec-channel-containment} later.

\begin{figure}
\begin{center}
\includegraphics[width=.6\columnwidth]{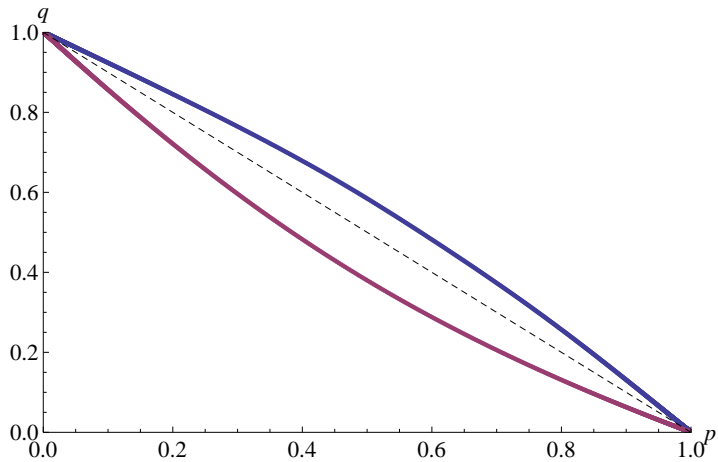}
\caption{
\label{fig:example-4d-channels}
Lower- (purple) and upper-bound (blue) curves 
for the 
mixing
probabilities $p$ and $q$ [defined in equations~\eqref{eqn-mixed-channel-E} and \eqref{eqn-mixed-channel-F}] for two random four-dimensional channels.
The solid curves are the bounds computed with equation~\eqref{eqn-final-bounds} and the dashed curve is $q=1-p$.
}
\end{center}
\end{figure}

\subsection{A pair of random four-dimensional channels}
\label{sec-example-random-4d-channels}

We randomly generated two four-dimensional channels each having four Kraus operators. (We do not list them here as they take up a lot of space.)
Figure~\ref{fig:example-4d-channels} shows the bounds.
We remark that the solid upper-bound curve computed with equation~\eqref{eqn-final-bounds} is not useful since it is above the $q=1-p$ line which is a trivial upper-bound.
Nevertheless, the lower-bound curve is useful since it allows only a narrow gap with the upper-bound line.

This example brings up an important point: in general, we should take the convex hull of an upper-bound curve as a refined upper-bound.

We also remark that the maximum number of cusps in the bounding curves is $16$ since the dimension of the channels is $4$.
However, they are not apparent in the figure.

\section
{Relation with distinguishability}
\label{sec-relation-distinguishability}

The diamond norm is related to the minimum error in discriminating between two quantum channels and is defined as
\begin{eqnarray*}
\lVert
\mathcal{G}
\rVert_\diamond
\triangleq
\max_\rho \lVert (I_\mathcal{K} \otimes \mathcal{G}) (\rho) \rVert_1
.
\end{eqnarray*}
Here, an ancillary Hilbert space $\mathcal{K}$ is introduced and 
$I_\mathcal{K}$ is the identity map acting on it.
The dimension of $I_\mathcal{K}$ is the same as the dimension of the Hilbert space of $\mathcal{G}$~\cite{Kitaev:1997}.
The minimum error in distinguishing between $\mathcal{E}$ and $\mathcal{F}$ is given by~\cite{Helstrom1976} 
\begin{eqnarray*}
P(\text{error})=
\frac{1}{2}
\left(1-
\frac{1}{2}
\lVert
\mathcal{E}-\mathcal{F}
\rVert_\diamond
\right).
\end{eqnarray*}

\subsection{Upper bound}

We can upper bound the diamond norm with the mixing probabilities $p$ and $q$ of our disguising problem as follows:
\begin{eqnarray}
\lVert
\mathcal{E}-\mathcal{F}
\rVert_\diamond
&=&
\lVert
(1-p)\mathcal{E}-(1-q)\mathcal{F}+p\mathcal{E}-q\mathcal{F}
\rVert_\diamond
\nonumber
\\
&=&
\lVert
q\mathcal{F}_\Delta-p\mathcal{E}_\Delta+p\mathcal{E}-q\mathcal{F}
\rVert_\diamond
\nonumber
\\
&\leq&
q \lVert \mathcal{F}_\Delta \rVert_\diamond + p \lVert \mathcal{E}_\Delta \rVert_\diamond + p \lVert \mathcal{E} \rVert_\diamond + q \lVert \mathcal{F} \rVert_\diamond
\nonumber
\\
&=&
2(p+q) ,
\label{eqn-diamond-p-q}
\end{eqnarray}
where the second line is due to equations~\eqref{eqn-mixed-channel-E} and \eqref{eqn-mixed-channel-F} with the disguising condition $\mathcal{E}'=\mathcal{F}'$ satisfied, the third line is due to the triangle inequality of the trace distance, and the fourth line comes from the fact that the trace norm of the channel output (a density matrix) is one.
Note that equation~\eqref{eqn-diamond-p-q} holds for any feasible $(p,q)$ satisfying $\mathcal{E}'=\mathcal{F}'$, not just the optimal $(p,q)$ trade-off curve.

When the two channels $\mathcal{E}$ and $\mathcal{F}$ are perfectly distinguishable, 
$\lVert
\mathcal{E}-\mathcal{F}
\rVert_\diamond = 2$.
On the other hand, in our disguising problem, $p+q=1$ is always achievable in equations~\eqref{eqn-mixed-channel-E} and \eqref{eqn-mixed-channel-F} since we can set $p=1-q$, $\mathcal{E}_\Delta=\mathcal{F}$, and $\mathcal{F}_\Delta=\mathcal{E}$.
Therefore, equation~\eqref{eqn-diamond-p-q} is tight in this case.

\subsection{Lower bound}

We focus on the case where $p=q$, which means that $\beta=1$.
In this case, from equation~\eqref{eqn-main-result-final-bounds}, the smallest $p=q$ satisfies
\begin{equation}
\label{eqn-diamond-norm-initial-p-q-T-Delta}
\frac{p}{1-p}=\frac{q}{1-q} \leq \lVert \ChSum(\Delta_+)\rVert ,
\end{equation}
where $\Delta_+$ is the positive subspace of $C_\mathcal{E}- C_\mathcal{F}$.

We divide $\Delta_+$ (of size $n^2 \times n^2$) into $n \times n$ blocks of equal size and
denote the $(i,j)$ block as $\Delta_{+ij}$.
Thus, 
$\Delta_+=\sum_{i,j=0}^{n-1} \ket{i}\bra{j} \otimes \Delta_{+ij}$, and 
it follows from the definition of $\ChSum$ in equation~\eqref{eqn-channel-sum-def} that
\begin{eqnarray}
\lVert \ChSum(\Delta_+)\rVert 
&=&
\left\lVert 
\sum_{i,j=0}^{n-1} \ket{i}\bra{j} \cdot \tr[\Delta_{+ij}]
\right\rVert
\nonumber 
\\
&=&
\max_{\ket{\phi}}
\sum_{i,j=0}^{n-1} 
\braket{\phi}{i}
\braket{j}{\phi}
\cdot \sum_{k=0}^{n-1} \bra{z_k}\Delta_{+ij} \ket{z_k}
\nonumber 
\\
&=&
\max_{\ket{\phi}}
\sum_{k=0}^{n-1} 
\bra{\phi}\otimes\bra{z_k}
\Delta_+
\ket{\phi}\otimes\ket{z_k}
\nonumber 
\\
&\leq&
\:
n
\max_{\ket{\phi},\ket{z}}
\bra{\phi}\otimes\bra{z}
\Delta_+
\ket{\phi}\otimes\ket{z} ,
\label{eqn-diamond-norm-special-max-Delta}
\end{eqnarray}
where on the second and third lines $\{\ket{z_k}\}$ is an orthonormal basis, and $\braket{\phi}{\phi}=\braket{z}{z}=1$.

Next, we consider the diamond norm:
\begin{eqnarray}
&&
\lVert
\mathcal{E}-\mathcal{F}
\rVert_\diamond
\nonumber
\\
&\geq&
\max_{\ket{\psi}}
\lVert
(I\otimes(\mathcal{E}-\mathcal{F}))(\ket{\psi}\bra{\psi})
\rVert
\nonumber
\\
\label{eqn-diamond-norm-expansion-1}
&=&
\max_{\ket{\sigma},\ket{\psi}}
|
\bra{\sigma}
(I\otimes(\mathcal{E}-\mathcal{F}))(\ket{\psi}\bra{\psi})
\ket{\sigma}
|
\triangleq Q ,
\end{eqnarray}
where the second line is due to the fact that the trace norm is no less than the spectral norm, and  
both the auxiliary system and the original system have dimension $n$.
Without loss of generality, using the Schmidt decomposition 
{\textcolor{mycolor2}{
on the auxiliary and original systems, 
}}%
we can express
\begin{eqnarray*}
\ket{\sigma}&=&\sum_{i=0}^{n-1} \gamma_i \ket{B_i} \ket{\sigma_i}
\\
\ket{\psi}&=&\sum_{i=0}^{n-1} \lambda_i \ket{B_i} \ket{\psi_i} ,
\end{eqnarray*}
where $\{\ket{B_i}\}$ is an orthonormal basis,
\begin{equation}
\label{eqn-diamond-norm-constraints-1}
\eqalign{
&\braket{\sigma_i}{\sigma_i}=\braket{\psi_i}{\psi_i}=1 \text{ for all $i=0,\ldots,n-1$, and}
\cr
&\sum_{i=0}^{n-1} |\gamma_i|^2=\sum_{i=0}^{n-1} |\lambda_i|^2=1.
}
\end{equation}
Continuing with equation~\eqref{eqn-diamond-norm-expansion-1},
the term to be maximized is equal to
\begin{eqnarray*}
&&
\left|
\sum_{i,j=0}^{n-1} 
\gamma_i^* \lambda_i \bra{\sigma_i}
(\mathcal{E}-\mathcal{F})(\ket{\psi_i}\bra{\psi_j})
\ket{\sigma_j} \lambda_j^* \gamma_j 
\right|
\\
&=&
\left|
\bra{v}
G
\ket{v}
\right|
\end{eqnarray*}
where
\begin{eqnarray*}
G\triangleq
\begin{bmatrix}
(\mathcal{E}-\mathcal{F})(\ket{\psi_0}\bra{\psi_0}) &(\mathcal{E}-\mathcal{F})(\ket{\psi_0}\bra{\psi_1})
&\cdots&
\\
(\mathcal{E}-\mathcal{F})(\ket{\psi_1}\bra{\psi_0})&\ddots
\\
\vdots
\\
&&&
(\mathcal{E}-\mathcal{F})(\ket{\psi_{n-1}}\bra{\psi_{n-1}})
\end{bmatrix}
\end{eqnarray*}
and
\begin{eqnarray*}
\ket{v}\triangleq
\begin{bmatrix}
\ket{\sigma_0} \lambda_0^* \gamma_0 
\\
\vdots
\\
\ket{\sigma_{n-1}} \lambda_{n-1}^* \gamma_{n-1} 
\end{bmatrix}.
\end{eqnarray*}
Note that $G$ is not a standard Choi matrix since there is no requirement that $\{\ket{\psi_i}\}$ is an orthonormal basis and also $\braket{v}{v}=\sum_i |\lambda_i|^2 |\gamma_i|^2$ may not be unity.

Continuing with equation~\eqref{eqn-diamond-norm-expansion-1},
we have
\begin{eqnarray*}
Q\triangleq
\max_{\left\{\ket{\sigma_i}\right\},\{\gamma_i\}}
\:
\max_{\left\{\ket{\psi_i}\right\},\{\lambda_i\}}
\left|
\bra{v}
G
\ket{v}
\right|
\end{eqnarray*}
subject to the constraints in equation~\eqref{eqn-diamond-norm-constraints-1}.
Since $\ket{\psi_i}=\ket{i}$ and $\lambda_i=1/\sqrt{n}$ satisfy the constraints,
\begin{eqnarray*}
Q
&\geq&
\max_{\ket{\tilde{v}}}
\frac{1}{n}
\left|
\bra{\tilde{v}} 
(C_\mathcal{E}- C_\mathcal{F})
\ket{\tilde{v}}
\right|
\\
&=&
\frac{1}{n}
\lVert 
C_\mathcal{E}- C_\mathcal{F}
\rVert
\\
&=&
\frac{1}{n}
\max(
\lVert 
\Delta_+
\rVert
,
\lVert 
\Delta_-
\rVert
) ,
\end{eqnarray*}
where the maximization is over any vector $\ket{\tilde{v}}$ with $\braket{\tilde{v}}{\tilde{v}}=1$ and 
$G$ with the substitution $\ket{\psi_i}=\ket{i}$ is equal to $C_\mathcal{E}- C_\mathcal{F}$.
Finally, note that
$\lVert 
\Delta_+
\rVert$ is larger than or equal to the maximization term in equation~\eqref{eqn-diamond-norm-special-max-Delta}.
Therefore,
\begin{equation}
\lVert
\mathcal{E}-\mathcal{F}
\rVert_\diamond
\geq
\frac{1}{n^2}
\lVert \ChSum(\Delta_+)\rVert 
\end{equation}
and combining with equation~\eqref{eqn-diamond-norm-initial-p-q-T-Delta}, the smallest $p=q$ must satisfy
\begin{equation}
\label{eqn-diamond-norm-lower-bound-final}
\frac{p}{n^2(1-p)}
\leq
\lVert
\mathcal{E}-\mathcal{F}
\rVert_\diamond .
\end{equation}

\subsection{Summary}
Using equations~\eqref{eqn-diamond-p-q} and \eqref{eqn-diamond-norm-lower-bound-final}, the smallest $p=q$
must satisfy
\begin{equation}
\label{eqn-diamond-norm-lower-upper-bounds-final}
\frac{p}{n^2(1-p)}
\leq
\lVert
\mathcal{E}-\mathcal{F}
\rVert_\diamond 
\leq
4p
.
\end{equation}
This shows that the disguising problem and the distinguishability problem are dual: when two channels are easy to distinguish (the diamond norm is large), it requires great effort to disguise one channel as the other ($p=q$ is large); and the reverse also holds.
Note that the lower bound of equation~\eqref{eqn-diamond-norm-lower-upper-bounds-final} is less than or equal to the upper bound since $p=q$ is at most $1/2$ due to the fact that any point on the line $q=1-p$ is feasible (cf. section~\ref{sec-problem-formulation-detail}).

\section{Other remarks}
\label{sec-remarks}

\subsection{Triangle inequality}

We apply the notion of triangle inequality to our mixing probabilities.
Suppose $\mathcal{E}$ and $\mathcal{F}$ are compatible with mixing probabilities $(p,q)$ and $\mathcal{G}$ and $\mathcal{F}$ with $(p',q')$, meaning that
\begin{eqnarray}
(1-p) \mathcal{E}(\rho)+p \mathcal{E}_\Delta(\rho) &=& (1-q) \mathcal{F}(\rho)+q \mathcal{F}_\Delta(\rho),
\label{app-triangle-E-eqn}
\\ 
(1-p') \mathcal{G}(\rho)+p' \mathcal{G}_\Delta(\rho) &=& (1-q') \mathcal{F}(\rho)+q' \mathcal{F}_\Delta'(\rho).
\label{app-triangle-G-eqn}
\end{eqnarray}
Note that the harmonizing channels $\mathcal{F}_\Delta$ and $\mathcal{F}_\Delta'$ are different in general.
We want to infer the distance profiles $(p'',q'')$ for $\mathcal{E}$ and $\mathcal{G}$ from the distance profiles $(p,q)$ and $(p',q')$.
To do this, we propose the following method: 
cross-multiply equations~\eqref{app-triangle-E-eqn} and \eqref{app-triangle-G-eqn} to make the coefficients of $\mathcal{F}$ equal and
add additional terms to the two resultant equations
to make the overall harmonizing channels on the right-hand sides equal.
The result is 
\begin{eqnarray*}
&&
(1-q')\left[(1-p) \mathcal{E}+p \mathcal{E}_\Delta\right] + (1-q) q' \mathcal{F}_\Delta' 
\\
&&\hspace{1cm}
= 
(1-q')\left[(1-q) \mathcal{F}+q \mathcal{F}_\Delta\right] + (1-q) q' \mathcal{F}_\Delta', \text{ and}
\\ 
&&
(1-q)\left[(1-p') \mathcal{G}+p' \mathcal{G}_\Delta\right] + (1-q')q \mathcal{F}_\Delta 
\\
&&\hspace{1cm}
=
 (1-q)\left[(1-q') \mathcal{F}+q' \mathcal{F}_\Delta'\right] + (1-q')q \mathcal{F}_\Delta,
\end{eqnarray*}
where we drop the dependence on $\rho$ for simpler notation and assume 
not both $q$ and $q'$ equal to $1$.
Thus, the two left-hand sides are equal, giving
\begin{eqnarray*}
&&
(1-q')\left[(1-p) \mathcal{E}+p \mathcal{E}_\Delta\right] + (1-q) q' \mathcal{F}_\Delta' 
\\
&&
\hspace{1cm}
= 
(1-q)\left[(1-p') \mathcal{G}+p' \mathcal{G}_\Delta\right] + (1-q')q \mathcal{F}_\Delta .
\end{eqnarray*}
This equation is interpreted as $\mathcal{E}$ occurring with probability 
$(1-q')(1-p)/(1-q q')$
and its harmonizing channel with probability 
\begin{equation}
\label{app-triangle-final-p}
p''=\frac{p(1-q')+(1-q)q'}{1-q q'} ,
\end{equation}
and $\mathcal{G}$ occurring with probability
$(1-q)(1-p')/(1-q q')$
and its harmonizing channel with probability 
\begin{equation}
\label{app-triangle-final-q}
q''=\frac{p'(1-q)+(1-q')q}{1-q q'} .
\end{equation}
With equations~\eqref{app-triangle-final-p}--\eqref{app-triangle-final-q}, given an achievable pair of mixing probabilities $(p,q)$ for $\mathcal{E}$ and $\mathcal{F}$ and another pair $(p',q')$ for $\mathcal{F}$ and $\mathcal{G}$, we can compute an achievable pair $(p'',q'')$ for $\mathcal{E}$ and $\mathcal{G}$.
Tracing out the entire distance profiles $(p,q)$ and $(p',q')$ produces an achievable region $(p'',q'')$.
Note that this region is not in general a curve.
Nevertheless, the bounding curve to this region can be regarded as a distance profile for $\mathcal{E}$ and $\mathcal{G}$.
This curve is certainly achievable but may not be optimal and so it represents an upper bound to the optimal trade-off distance profile.

Note that even though 
we have 
the assumption $q q'<1$, we can still obtain the end points.
In particular, when $(p,q)=(0,1)$ and $q'<1$, we get $(p'',q'')=(0,1)$; 
when $(p',q')=(0,1)$ and $q<1$, we get $(p'',q'')=(1,0)$.

As an example, consider $\mathcal{E}$ and $\mathcal{F}$ given in 
equations~\eqref{eqn-example-bit-flip} and \eqref{eqn-example-phase-flip} and 
\begin{equation*}
\mathcal{G}(\rho) = (1-c) I_2 \rho I_2 + c (XZ) \rho (ZX)
\end{equation*}
with $a=b=c=0.2$.
Figure~\ref{fig:example-triangle} shows the achievable region for $\mathcal{E}$ and $\mathcal{G}$ obtained by 
equations~\eqref{app-triangle-final-p} and \eqref{app-triangle-final-q} together with the optimal curve obtained in section~\ref{sec-example-flip-channels}.
Note that due to symmetry, any pair of $\mathcal{E}$, $\mathcal{F}$, and $\mathcal{G}$ has the same optimal trade-off curve.

\begin{figure}
\begin{center}
\includegraphics[width=.6\columnwidth]{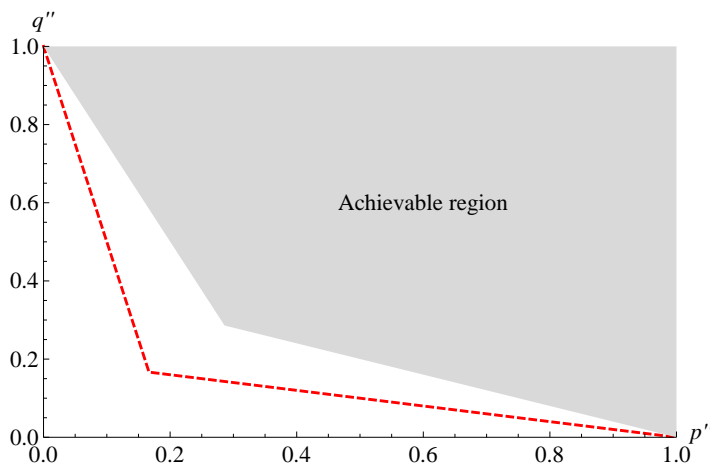}
\caption{
\label{fig:example-triangle}
Achievable region for $\mathcal{E}$ and $\mathcal{G}$ obtained by 
equations~\eqref{app-triangle-final-p} and \eqref{app-triangle-final-q} using the mixing probabilities of $\mathcal{E}$ and $\mathcal{F}$ and the mixing probabilities of $\mathcal{F}$ and $\mathcal{G}$.
The dashed curve (red) is the optimal trade-off.
}
\end{center}
\end{figure}

We now consider the triangle inequality for two scalar distances derived from our 2-dimensional measure: (i) the minimum of $p$ subject to $p=q$ and (ii) the minimum of $p+q$.
For case (i), suppose that the minimum for $\mathcal{E}$ and $\mathcal{F}$ is $p_0$ with $q_0=p_0$ and the minimum for $\mathcal{F}$ and $\mathcal{G}$ is $p_0'$ with $q_0'=p_0'$.
Substituting these four parameters into equations~\eqref{app-triangle-final-p} and \eqref{app-triangle-final-q} gives 
\begin{equation*}
p''=\frac{p_0(1-p_0')+(1-p_0)p_0'}{1-p_0 p_0'}=q'' .
\end{equation*}
This means that the $p=q$ condition is preserved.
Since $(1-p_0')(1-p_0 p_0') \leq 1$ and $(1-p_0)(1-p_0 p_0') \leq 1$, 
we have $p'' \leq p_0 + p_0'$.
Finally, since $p''$ is only an achievable upper bound to the optimal mixing probability for $\mathcal{E}$ and $\mathcal{G}$, the triangle inequality is satisfied.

For case (ii), suppose that the minimum for $\mathcal{E}$ and $\mathcal{F}$ is $p_0+q_0$ 
and the minimum for $\mathcal{F}$ and $\mathcal{G}$ is $p_0'+q_0'$ .
Substituting these four parameters into equations~\eqref{app-triangle-final-p} and \eqref{app-triangle-final-q} gives 
\begin{equation*}
p''+q''=\frac{(p_0+q_0)(1-q_0') + (p_0'+q_0')(1-q_0)}{1-q_0 q_0'}.
\end{equation*}
Using the same argument as in case (i), we have $p''+q'' \leq (p_0+q_0) + (p_0'+q_0')$ and the triangle inequality is satisfied.

\subsection
{Channel containment}
\label{sec-channel-containment}

Given two quantum channels $\mathcal{E}$ and $\mathcal{F}$, we introduce the notion that
$\mathcal{E}$ contains $\mathcal{F}$ 
if
$\mathcal{F}$ is
part of a mixture of $\mathcal{E}$: 
\begin{equation}
\label{eqn-containment-condition}
\mathcal{E}(\rho)
=
(1-q) \mathcal{F}(\rho) + q  \mathcal{F}_\Delta (\rho)
\end{equation}
where $\mathcal{F}_\Delta$ is some quantum channel and we require that $q<1$ to make the containment of $\mathcal{F}$ non-trivial.
To see if equation~\eqref{eqn-containment-condition} holds, we convert it to the Choi representation and proceed as in equation~\eqref{eqn-channels-difference1} with $p=0$.
The question becomes for what values of $q$ is $C_\mathcal{E}-(1-q) C_\mathcal{F}=q C_{\mathcal{F}_\Delta}$ PSD.
This is the same as asking whether $\Delta_-=0$ in equation~\eqref{eqn-channel-decomposition1}.
In general, we fix a value of $q$ and perform eigen-decomposition to see if $\Delta_-=0$.
But for the case that $C_\mathcal{F}$ is invertible, we can find the minimum $q$ by equating $1-q$ with the minimum eigenvalue of
$C_\mathcal{F}^{-1/2} C_\mathcal{E} C_\mathcal{F}^{-1/2}$.

Note that 
$C_{\mathcal{F}_\Delta}$ is automatically trace-preserving, since
$q\ChSum(C_{\mathcal{F}_\Delta})=\ChSum(C_\mathcal{E}-(1-q) C_\mathcal{F})=I-(1-q)I=qI$ as both $\mathcal{E}$ and $\mathcal{F}$ are trace-preserving.

\subsection{Composition of quantum channels}

Suppose that
$\mathcal{E}_i$ and $\mathcal{F}_i$ can be made compatible with $p_i$ and $q_i$ according to the processing in 
equations~\eqref{eqn-mixed-channel-E}--\eqref{eqn-mixed-channel-F}, where $i=1,2$.
This means that
$\mathcal{E}_i'(\rho)
=\mathcal{F}_i'(\rho)$ or
\begin{equation*}
(1-p_i) \mathcal{E}_i(\rho) + p_i  \mathcal{E}_{\Delta i} (\rho)
=
(1-q_i) \mathcal{F}_i(\rho) + q_i  \mathcal{F}_{\Delta i} (\rho)
\end{equation*}
for all density matrices $\rho$ and $i=1,2$.
Then the composed quantum channels 
$\mathcal{E}_2 \circ \mathcal{E}_1$ and
$\mathcal{F}_2 \circ \mathcal{F}_1$
can also be made compatible with $p=p_1+p_2-p_1 p_2$ and $q=q_1+q_2-q_1 q_2$:
\begin{equation}
(1-p) \mathcal{E}_2 \circ \mathcal{E}_1(\rho) + p  \mathcal{E}_\Delta (\rho)
=
(1-q) \mathcal{F}_2 \circ \mathcal{F}_1(\rho) + q  \mathcal{F}_\Delta (\rho) .
\end{equation}

To show this, note that
since 
{\textcolor{mycolor2}{
$\mathcal{E}_i(\rho)$ and $\mathcal{F}_i(\rho)$ are density matrices for any $\rho$, 
}}%
the following holds:
\begin{equation}
\mathcal{E}_2'(\mathcal{E}_1'(\rho))
=
\mathcal{F}_2'(\mathcal{F}_1'(\rho)) .
\end{equation}
Expansion of the LHS gives
\begin{eqnarray}
\mathcal{E}_2'\circ \mathcal{E}_1'
&=
&
(1-p_1)(1-p_2)\mathcal{E}_2 \circ \mathcal{E}_1
\label{eqn-composition-composed-channel-E}
+
p_1(1-p_2) \mathcal{E}_2 \circ \mathcal{E}_{\Delta 1}
\\
&&
+
(1-p_1)p_2 \mathcal{E}_{\Delta 2} \circ \mathcal{E}_1
\nonumber
+
p_1 p_2\mathcal{E}_{\Delta 2} \circ \mathcal{E}_{\Delta 1} .
\nonumber
\end{eqnarray}
We can readily see that the first term is the original composed channel with 
mixing
probability $1-p=(1-p_1)(1-p_2)$ and the sum of the last three terms represents the 
harmonizing channel
$\mathcal{E}_\Delta$ with 
mixing
probability $p=p_1+p_2-p_1 p_2$.
Together with a similar argument for $\mathcal{F}_2'\circ \mathcal{F}_1'$ proves the claim.

We can also argue that 
the composed quantum channels 
$\mathcal{E}_2 \circ \mathcal{E}_1$ and
$\mathcal{F}_2 \circ \mathcal{F}_1$
can be made compatible with $p=p_1+p_2$ and $q=q_1+q_2$,
by breaking up the first term of equation~\eqref{eqn-composition-composed-channel-E} and allocating the portion $p_1 p_2 \mathcal{E}_2 \circ \mathcal{E}_1$ to the 
harmonizing channel
(and similarly for $\mathcal{F}$).

\section{Application to quantum cryptography}
\label{sec-application}

The disguising condition $\mathcal{E}'=\mathcal{F}'$ can be used to upper bound the key generation rate in quantum cryptography~\cite{Bennett1984,Ekert1991}.
The intuitive idea is that the more easily the eavesdropper's channel can be disguised as the legitimate user's channel, the smaller is the amount of the generated key.
We establish this idea quantitatively relating the mixing probability and the key generation rate.

In quantum cryptography with one-way forward communications, Alice repeatedly sends a quantum state $\rho_a$ to Franky to establish a secret key, where $a=(a_0, a_1)$ and $a_0$ ($a_1$) is Alice's raw key basis (value) chosen independently between different transmissions.
After the reception of the sequence of states by Franky, Alice sends classical information (including basis information, error correction information, and privacy amplification information) to him in order to correct bit errors and remove any information the eavesdropper Eve may have on the final key.
Suppose that Eve launches a collective attack~\cite{Biham1997} which means that she applies the same unitary transformation $U$ to each state sent by Alice (with sufficient ancillas).
Thus, Franky's channel and Eve's channel are given as follows:
\begin{eqnarray}
\mathcal{F}(\rho_a)&=&\tr_{\tt E}( U (\rho_a \otimes \ket{0}\bra{0}) U^\dag )
\\
\mathcal{E}(\rho_a)&=&\tr_{\tt F}( U (\rho_a \otimes \ket{0}\bra{0}) U^\dag )
\end{eqnarray}
where we assume without loss of generality (w.l.o.g.) that the entire Hilbert space is divided into two systems {\tt E} and {\tt F}.
Furthermore, for simplicity and w.l.o.g., we assume that {\tt E} and {\tt F} have the same dimensions $n$ (which can be assured by padding zeros as needed).

A key rate upper bound is given by the classical secret key capacity formula~\cite{Csiszar1978}
\begin{equation}
\label{eqn-QKD-key-rate}
R=\sup_{\substack{U\leftarrow A\\V\leftarrow U}} I(U;F|V)-I(U;E|V) .
\end{equation}
Note that the use of the classical formula is valid since when considering the upper bound, we can assume one particular strategy of Eve, which is to measure her quantum states separately.
Here, $A$ is Alice's random variable holding the raw key $a$, $F$ and $E$ are the Franky's and Eve's random variables holding the measurement outcomes $f$ and $e$ respectively.

We consider the upper bound of the key rate for the case where the disguising condition is
\begin{equation}
\label{eqn-qkd-application-disguising-condition1}
\mathcal{F}(\rho)
=
(1-p) \mathcal{E}(\rho) + p  \mathcal{E}_\Delta (\rho) .
\end{equation}
That is, $1-p$ fraction of Franky's channel is Eve's channel.
Thus, one would expect that only the remaining fraction of $p$ could be used to generate a secret key and the key rate would be on the order of $p$.

For simplicity of discussion, instead of equation~\eqref{eqn-QKD-key-rate}, we bound the key rate expression without the processing:
\begin{equation}
\label{eqn-QKD-key-rate2}
R'=I(A;F)-I(A;E) .
\end{equation}

%

Suppose that Bob's POVM is $\{ M_{f,a_0}\}$ where $\sum_f M_{f,a_0}=I$ and 
it is dependent on the raw key basis $a_0$.
Applying this POVM to
equation~\eqref{eqn-qkd-application-disguising-condition1}
produces classical probability distributions
\begin{eqnarray*}
P_{AF}(a,f)&=&P_A(a) \tr ( M_{f,a_0} \mathcal{F}(\rho_a))
\\
P_{AE}(a,f)&=&P_A(a) \tr ( M_{f,a_0} \mathcal{E}(\rho_a))
\\
P_{AE_\Delta}(a,f)&=&P_A(a) \tr ( M_{f,a_0} \mathcal{E}_\Delta(\rho_a))
\end{eqnarray*}
which are related by
\begin{equation}
\label{eqn-qkd-application-prob-dist-AF-1}
P_{AF}(a,f)=
(1-p) P_{AE}(a,f)
+
p P_{AE_\Delta}(a,f).
\end{equation}

Note that this relation can be explained by
the following hypothetical probability distribution
\begin{equation*}
P_{AFZ}(a,f,z)
=
\begin{cases}
(1-p) P_{AE}(a,f), & \mbox{if } z=0
\\
p P_{AE_\Delta}(a,f), & \mbox{if } z=1 
\end{cases}
\end{equation*}
in that
$\sum_{z=0,1} P_{AFZ}(a,f,z)$ is equal to
equation~\eqref{eqn-qkd-application-prob-dist-AF-1}.

Now, we bound the first term of equation~\eqref{eqn-QKD-key-rate2} as follows:
\begin{eqnarray*}
I(A;F) &\leq& I(A;FZ)
\\
&=&
I(A;F|Z)+I(A;Z)
\\
&=&
(1-p) I(A;F|z=0) + p I(A;F|z=1)
\\
&=&
(1-p) I(A;E) + p I(A;E_\Delta)
\end{eqnarray*}
where on the second line we have $I(A;Z)=0$ since $P_{AZ}(a,z)=P_A(a) P_Z(z)$.
Therefore, using equation~\eqref{eqn-QKD-key-rate2}, the key rate is bounded as
\begin{eqnarray}
R'
&\leq&
(1-p) I(A;E) + p I(A;E_\Delta) - I(A;E)
\nonumber
\\
&=&
p [ I(A;E_\Delta) - I(A;E) ]
\nonumber
\\
&\leq&
p \log_2 n
\label{eqn-QKD-final-upper-bound-formula}
\end{eqnarray}
where the last inequality is due to Holevo-Schumacher-Westmoreland channel capacity theorem~\cite{Holevo:1998,Schumacher:1997} and the fact that the maximum entropy for an $n$-dimensional state is $\log_2 n$.
A similar analysis can be applied to the original key rate expression in equation~\eqref{eqn-QKD-key-rate} to obtain the same upper bound in equation~\eqref{eqn-QKD-final-upper-bound-formula}.

\section{Conclusions}
\label{sec-conclusions}

The disguising problem tries to make two quantum channels identical, which is the reverse of the distinguishability problem which tries to maximize their difference in the measurement statistics in order to discriminate between them.
Indeed, we showed that the two problems are related by proving that a certain combination of the mixing probabilities of the disguising problem upper bounds the diamond norm of the distinguishability problem.
We also showed that the triangle equality holds for two scalar distances derived from the mixing probabilities.

Conventional measures on quantum channels are mostly based on trace distance or fidelity which are both concepts derived from measuring the distance between quantum states.
In this paper, we propose a new measure 
genuinely
for quantum channels.
Note that the application of our measure to quantum states is possible but the result would be rather trivial since there is no more the need of making sure the 
harmonizing channels
satisfy the TP condition for a linear map (which is ensured by the addition of $X$ in equations~\eqref{eqn-delta-channels1-F} and \eqref{eqn-delta-channels1-E}).
This extra condition makes the calculation of our measure for quantum channels more difficult.
Nevertheless, we obtain analytical lower- and upper-bounds which generate curves that are close 
{\textcolor{mycolor2}{
to each other 
}}%
in many cases.

Our measure is based on the notion of minimizing the probabilities of channel mixing, which can be viewed as the costs for an channel intervener to make two channels the same.
We show how these costs are linked to the key generation rate in quantum key distribution.
The investigation of how these costs are linked to other quantum information processing tasks 
is a topic for future research.
Also, open problems include
efficient/approximate computation of our measure, and the effect of extending the Hilbert space dimensions of the channels by including ancillary systems (i.e., $\mathcal{E}$ becomes $\mathcal{E}\otimes I$) in the calculation of our measure.

\section*{Acknowledgments}%
We thank Masahito Hayashi, Seung-Hyeok Kye, and Jonathan Oppenheim for enlightening discussion.
This work is supported in part by
RGC under Grant No. 700709P and 700712P of the HKSAR Government.

\appendix

\section{Useful results related to quantum channels}
\label{app-preliminary}

This section discusses the tools and definitions related to quantum channels.
A quantum channel 
is a linear map that is completely-positive (CP)
and
trace-preserving (TP).

\begin{theorem}
\label{thm-Choi-CP}
{\rm
(Choi's theorem~\cite{Choi1975:CPMap})
Given a linear map $\mathcal{E}$ and its Choi matrix 
$C_\mathcal{E}$,
$C_\mathcal{E}$ is PSD if and only if $\mathcal{E}$ is a completely-positive map.
}
\end{theorem}
When
$C_\mathcal{E}$ is Hermitian, $\mathcal{E}$ can also be represented in
the operator-sum form:
\begin{equation}
\label{eqn-channel-operator-sum-form}
\mathcal{E}(\rho)=\sum_i \lambda_i E_i \rho E_i^\dag, 
\end{equation}
where $\lambda_i \in \IR$ and $E_i \in \IC^{n,n}$.
This can be seen by taking some decomposition (e.g., eigen-decomposition) of $C_\mathcal{E}$ to be
\begin{equation}
\label{eqn-choi-decomposition}
C_\mathcal{E}=\sum_i \lambda_i \ket{E_i} \bra{E_i} 
\end{equation}
where $\lambda_i \in \IR$ and $\ket{E_i} \in \IC^{n^2,1}$, and 
rearranging the vector $\ket{E_i}$ into the square matrix $E_i$ as follows:
\begin{eqnarray}
\label{eqn-vector-form}
\ket{E_i} \equiv
\begin{bmatrix}
E_i(1,1)\\
E_i(2,1)\\
\vdots\\
E_i(n,1)\\
E_i(1,2)\\
\vdots\\
E_i(n,2)\\
\vdots\\
E_i(n,n)
\end{bmatrix}
=\operatorname{vec}(E_i)
\end{eqnarray}
where $E_i(k,l)$ is the $(k,l)$ entry of $E_i$, and the $\operatorname{vec}$ operator creates a vector by stacking the columns of its operand (see, e.g., Ref.~\cite{Horn1994}).
The dimension of $E_i$ is $n \times n$ and
the dimension of $C_\mathcal{E}$ is $n^2 \times n^2$.
Note that the operator-sum form is not unique; there can be more than one such form corresponding to the same Choi matrix.
\begin{observation}
{\rm
The Choi matrix of any channel given in the operator-sum form of equation~\eqref{eqn-channel-operator-sum-form} can be constructed by using
equation~\eqref{eqn-choi-matrix} or
equation~\eqref{eqn-choi-decomposition}.
The equivalence of these two ways can be checked easily by direct expansion.
}
\end{observation}

The next Theorem follows directly from the definition of the Choi matrix in equation~\eqref{eqn-choi-matrix}.
\begin{theorem}
\label{thm-equiv-channels}
{\rm
$\mathcal{E}(\rho) = \mathcal{F}(\rho)$ for all density matrices $\rho$ if and only if $C_\mathcal{E}=C_\mathcal{F}$.
}
\end{theorem}

\begin{definition}
\label{def-channel-sum}
{\rm
We define the {\em channel sum} for linear map $\mathcal{E}$ with a Hermitian Choi matrix as
\begin{equation}
\label{eqn-def-channel-sum}
T_\mathcal{E} \triangleq \sum_i \lambda_i E_i^\dag E_i,
\end{equation}
which has dimension $n \times n$.
}
\end{definition}%
The next lemma shows how to obtain the channel sum from channel outputs directly.
\begin{lemma}
\label{lemma-channel-sum-of-Choi-equiv}
{\rm
\begin{eqnarray*}
T_\mathcal{E} &=& \sum_{i,j=0}^{n-1} \ket{i}\bra{j} \cdot \tr[\mathcal{E}(\ket{j}\bra{i})]
\\
&=&\tr_2^t ( C_\mathcal{E} )
\end{eqnarray*}
where $\tr_2$ is the partial trace over the second system and $t$ represents transpose.
Note that here, it does not matter whether the transpose is taken after or before the partial trace.
}
\end{lemma}
\begin{proof}
The $(i,j)$ element of $S_\mathcal{E}$ defined in equation~\eqref{eqn-def-channel-sum} is
\begin{eqnarray*}
\bra{i}T_\mathcal{E}\ket{j} &=& \bra{i} \sum_k \lambda_k E_k^\dag E_k \ket{j}\\
&=&\tr \left[ \sum_k \lambda_k E_k \ket{j} \bra{i} E_k^\dag \right]\\
&=&\tr \left[ \mathcal{E}(\ket{j}\bra{i}) \right].
\end{eqnarray*}
\end{proof}
Therefore, $T_\mathcal{E}$ is 
independent of the operator-sum form of $\mathcal{E}$ and is dependent only on the Choi matrix.
This allows us to define the channel sum function 
\begin{eqnarray}
\label{eqn-channel-sum-def-appendix}
\ChSum(C_\mathcal{E})&:=&\tr_2^t ( C_\mathcal{E} )
\\
&=&
T_\mathcal{E} \nonumber
\end{eqnarray}
where we used Lemma~\ref{lemma-channel-sum-of-Choi-equiv}.
The introduction of $\ChSum$
facilitates 
the
discussion of the channel sum with reference to only the Choi matrix.
\begin{lemma}
\label{lemma-TP}
{\rm
A linear map $\mathcal{E}$ is trace-preserving if and only if 
$\ChSum(C_\mathcal{E})=I$.
}
\end{lemma}
\begin{proof}
For $\mathcal{E}$ to be trace-preserving, the following must hold for all density matrices $\rho$:
\begin{eqnarray*}
\tr(\rho)
&=&
\tr[ \mathcal{E}(\rho)]
\\
&=&
\tr( T_\mathcal{E} \rho)
\end{eqnarray*}
where the last equality is due to equation~\eqref{eqn-def-channel-sum}.
Since this holds for all $\rho$, $T_\mathcal{E}=I$.
Then, using equation~\eqref{eqn-channel-sum-def-appendix},
$\ChSum(C_\mathcal{E})=I$.

The proof for the other direction is obvious.
\end{proof}

We remark that the eigenvalues of a Choi matrix $C$ and of its channel sum $\ChSum(C)$ are in general not the same.
Nevertheless, they have the same trace.
\begin{corollary}
\label{cor-same-trace}
{\rm
$\tr (C) = \tr (\ChSum(C))$ for 
{\textcolor{mycolor2}{
any 
}}%
Choi matrix $C$.
}
\end{corollary}
Thus, the trace of the Choi matrix of a quantum channel is its Hilbert space dimension $n$ because the channel sum of a quantum channel is $I_n$.
This corollary will be useful when we consider the lower bound of the channel distance.
\begin{corollary}
\label{cor-S-linear}
{\rm
$C_\mathcal{E}+C_\mathcal{F}=C_{\mathcal{E}+\mathcal{F}}$ and
$\ChSum(C_{\mathcal{E}+\mathcal{F}})=\ChSum(C_\mathcal{E})+\ChSum(C_\mathcal{F})$ 
for linear maps $\mathcal{E}$ and $\mathcal{F}$.
}
\end{corollary}

Also, note that it can easily be checked that if $C$ is PSD, $\ChSum(C)$ is also PSD.

\section{Channel specification for section~\ref{sec-example-random-qubit-channels}}
\label{app-example-random-qubit-channels}

The two channels are 
$\mathcal{E}(\rho)=\sum_{i=1}^4 E_i \rho E_i^\dag$ 
and
$\mathcal{F}(\rho)=\sum_{i=1}^4 F_i \rho F_i^\dag$,
where
\begin{equation*}
\begin{array}{rl}
E_1 &=
\left[
\begin{array}{cc}
 -0.504828 & -0.331944 \\
 -0.0133105 & 0.295026
\end{array}
\right],
\smallskip
\\
E_2 &=
\left[
\begin{array}{cc}
 0.419485 & 0.158018 \\
 0.330761 & 0.0616354
\end{array}
\right],
\smallskip
\\
E_3 &=
\left[
\begin{array}{cc}
 0.464696 & 0.251826 \\
 -0.312786 & 0.165248
\end{array}
\right],
\smallskip
\\
E_4 &=
\left[
\begin{array}{cc}
 0.160149 & -0.346665 \\
 -0.346665 & 0.750403
\end{array}
\right],
\smallskip
\\
F_1 &=
\left[
\begin{array}{cc}
 -0.20917 & -0.248828 \\
 0.382771 & -0.451866
\end{array}
\right],
\smallskip
\\
F_2 &=
\left[
\begin{array}{cc}
 -0.62412 & -0.425856 \\
 0.286902 & -0.0613943
\end{array}
\right],
\smallskip
\\
F_3 &=
\left[
\begin{array}{cc}
 0.216184 & -0.422341 \\
 -0.403389 & 0.451605
\end{array}
\right],
\smallskip
\\
F_4 &=
\left[
\begin{array}{cc}
 0.236514 & 0.269256 \\
 0.269256 & 0.306531
\end{array}
\right]
.
\end{array}
\end{equation*}

\section*{References}

\bibliographystyle{iopart-num}

\bibliography{paperdb}

\end{document}